\documentclass[a4paper,USenglish,cleveref, autoref, thm-restate]{lipics-v2019}
\nolinenumbers

\title{
Synchronizing Deterministic Push-Down Automata Can Be Really  Hard}

\titlerunning{Synchronizing DPDAs Can Be Really  Hard} 

\author{Henning Fernau}{Universität Trier, Fachbereich IV, Informatikwissenschaften, 54296 Trier, Germany} {fernau@uni-trier.de}{https://orcid.org/0000-0002-4444-3220}{}

\author{Petra Wolf}{Universität Trier, Fachbereich IV, Informatikwissenschaften, 54296 Trier, Germany \and \url{https://www.wolfp.net/}}{wolfp@uni-trier.de}{https://orcid.org/0000-0003-3097-3906}{DFG project  
	FE 560/9-1}

\author{Tomoyuki Yamakami}
{University of Fukui, Faculty of Engineering,   3-9-1 Bunkyo, Fukui 910-8507,
	Japan }{tomoyukiyamakami@gmail.com}{}{}

\authorrunning{H.\ Fernau, P.\ Wolf and T.\ Yamakami} 

\Copyright{Henning Fernau, Petra Wolf, Tomoyuki Yamakami} 



\category{} 

\supplement{}

\EventEditors{John Q. Open and Joan R. Access}
\EventNoEds{2}
\EventLongTitle{42nd Conference on Very Important Topics (CVIT 2016)}
\EventShortTitle{CVIT 2016}
\EventAcronym{CVIT}
\EventYear{2016}
\EventDate{December 24--27, 2016}
\EventLocation{Little Whinging, United Kingdom}
\EventLogo{}
\SeriesVolume{42}
\ArticleNo{23}

\usepackage{float}
\usepackage{todonotes}
\usepackage{amsmath}

\newcommand{\NP}{\textsf{NP}}
\newcommand{\PSPACE}{\textsf{PSPACE}}
\newcommand{\NPSPACE}{\textsf{NPSPACE}}
\newcommand{\NLOGSPACE}{\textsf{NLOGSPACE}}
\newcommand{\PTIME}{\textsf{P}}

\newcommand{\EXP}{\textsf{EXPTIME}}

\newcommand{\mto}[1][w]{\ensuremath{\overset{#1}{\longrightarrow}}}
\DeclareMathOperator{\lang}{\mathcal{L}}
\newcommand{\cerny}{{\v{C}}ern{\'{y}}}
\usepackage{soul}
\usepackage{MnSymbol,wasysym}
\usepackage{booktabs}
\usepackage{tikz}
\usetikzlibrary{arrows,automata}
\usepackage{apxproof}

\begin{document}

\maketitle

\begin{abstract}
The question if a deterministic finite automaton admits a software reset in the form of a so-called synchronizing word can be 
answered in polynomial time. In this paper, we extend this algorithmic question to deterministic automata beyond finite automata.
We prove that the question of synchronizability becomes undecidable even when looking at deterministic one-counter automata.
This is also true for another classical mild extension of regularity, namely that of deterministic one-turn push-down automata.
However, when we combine both restrictions, we arrive at scenarios with a \PSPACE-complete (and hence decidable) synchronizability problem.
Likewise, we arrive at a decidable synchronizability problem for (partially) blind deterministic counter automata.

There are several interpretations of what \emph{synchronizability} should mean for deterministic push-down automata. This is depending on the role of the stack: should it be empty on synchronization, should it be always the same or is it arbitrary?
For the automata classes studied in this paper, the complexity or decidability status of the synchronizability problem is mostly independent of this technicality, but we also discuss one class of automata where this makes a difference.
\end{abstract}

\section{Introduction}

The classical \emph{synchronization problem} asks, for a given deterministic finite automaton (DFA), if there exists a \emph{synchronizing word}, i.e., an input that brings all states of the automaton to a single state. While this problem is solvable in polynomial time~\cite{cerny1964,Vol2008,San2005}, many variants, such as synchronizing only a subset of states~\cite{San2005}, or synchronizing only into a specified subset of states~\cite{DBLP:journals/ipl/Rystsov83}, 
or synchronizing a partial automaton without taking an undefined transition~\cite{DBLP:journals/mst/Martyugin14}, are \PSPACE-complete. Restricting the length of a potential synchronizing word of some DFA by an integer parameter in the input also yields a harder problem, namely the \NP-complete short synchronizing word problem~\cite{Rys80,DBLP:journals/siamcomp/Eppstein90}.
The field of synchronizing automata has been intensively studied over the last years also in attempt to verify the famous \cerny\ conjecture, claiming that every synchronizable DFA admits a synchronizing word of quadratic length in the number of states~\cite{cerny1964,DBLP:journals/jalc/Cerny19,DBLP:journals/eik/Starke66b,DBLP:journals/jalc/Starke19}. We are far from solving this combinatorial question, as the currently best upper bound on this length is only cubic~\cite{DBLP:journals/jalc/Shitov19, DBLP:conf/stacs/Szykula18}.
For more on synchronization of DFAs and the \cerny\ conjecture, we refer to the surveys~\cite{Vol2008,beal_perrin_2016,JALC20}.

The idea of bringing an automaton to a well-defined state by reading a word, starting from any state, can be seen as implementing a software reset. This is why a synchronizing word is also sometimes called a \emph{reset word}. But this very idea is obviously not restricted to finite automata. 
In this work, we want to move away from deterministic finite automata to more general deterministic push-down automata.  
What should a synchronizing word mean in this context? Mikami and Yamakami first studied in~\cite{MY20} three different models, depending on requirements of the stack contents when a word~$w$ drives the automaton into a synchronizing state, irrespectively of the  state where processing~$w$ started: we could require  at the end (a) that the stack is always empty; or (b) that the stack contents is always the same (but not necessarily empty); or (c) that the
stack contents is completely irrelevant upon entering the synchronizing state. 
They demonstrated in~\cite{MY20} some upper and lower bounds on the maximum length of the shortest synchronizing word for those three models of push-down automata, dependent on the stack height. 
Here, we study these three models from a complexity-theoretic perspective.  However, as we show in our first main result, synchronizability becomes undecidable when asking about synchronizability in any of the stack models. 
Clearly, by restricting the length of a potential synchronizing word of some DPDA by an integer parameter (given in unary),
we can observe that the corresponding synchronization problems all become \NP-complete, as the hardness is trivially inherited from what is known about DFA synchronizability. Therefore, we will not consider such length-bounded problem variants any further in this paper. Yet, it remains interesting to observe that with DFAs, introducing a length bound on the synchronizing word means an increase of complexity, while for DPDAs, this introduction means dropping from undecidability close to feasibility. 
Beside general DPDAs, we will study these stack model variants of synchronization for sub-classes of DPDAs such as deterministic counter automata (DCA), deterministic  (partially) blind automata and finite-turn variants of DPDAs and DCAs.
In~\cite{vpda-crossref}, further restricted sub-classes of DPDAs, such as visibly and very visibly deterministic push-down and counter automata are considered. There, all considered cases are in \EXP\ and even membership in \PTIME\ and \PSPACE\ is observed, contrasting our undecidability results here.

Closest to the problems studied in our paper comes the  work of Chistikov et al.~\cite{DBLP:journals/jalc/ChistikovMS19} reviewed in the following, as their automaton model could be viewed as a special case of push-down automata, related to \emph{input-driven pushdown automata}~\cite{DBLP:conf/icalp/Mehlhorn80} which later became popular as \emph{visibly push-down automata}~\cite{DBLP:conf/stoc/AlurM04}. 
In~\cite{DBLP:journals/jalc/ChistikovMS19}, the synchronization problem for so-called \emph{nested word automata} (NWA) has been studied, where the concept of synchronization has been generalized to bringing all states to one single state such that for all runs the stack is empty (or in its start configuration) after reading the synchronizing word. In this setting the synchronization problem is solvable in polynomial time, whereas the short synchronizing word problem is \PSPACE-complete (here, the length bound is given in binary) and the question of synchronizing from or into a subset is \EXP-complete.

The DFA synchronization problem has been generalized in the literature to other automata models including infinite-state systems with infinite branching such as weighted and timed automata~\cite{DBLP:conf/fsttcs/0001JLMS14,DBLP:phd/hal/Shirmohammadi14} or register automata~\cite{babari2016synchronizing}. 
For instance, register automata are infinite state systems where a state consists of a control state and register contents. A synchronizing word for a register automaton brings all (infinitely many) states to the same state (and same register content). The synchronization problem for deterministic register automata (DRA) is \PSPACE-complete and \NLOGSPACE-complete for DRAs with only one register.

Finally, we want to mention that the term \emph{synchronization of push-down automata} has already some occurrences in the literature, i.e., in ~\cite{caucal2006synchronization,DBLP:journals/mst/ArenasBL11}, but here the term \emph{synchronization} refers to some relation of the input symbols to the stack behavior~\cite{caucal2006synchronization} or to reading different words in parallel~\cite{DBLP:journals/mst/ArenasBL11} and is not to be confused with our notion of synchronizing states.

We are presenting an overview on our results at the end of the next section, where we introduce our notions more formally.

\section{Definitions}

We refer to the empty word as $\epsilon$.
For a finite alphabet $\Sigma$ we denote with $\Sigma^*$ the set of all words over $\Sigma$ and with $\Sigma^+ = \Sigma\Sigma^*$ the set of all non-empty words.
For $i\in \mathbb{N}$ we set $[i] = \{1, 2, \ldots, i\}$.
For $w \in \Sigma^*$ we denote with $|w|$ the length of $w$, with $w[i]$ for $i \in [|w|]$ the $i$'th symbol of $w$, and with $w[i..j]$ for $i, j \in [|w|]$ the factor $w[i]w[i+1] \ldots w[j]$ of $w$. We call $w[1..i]$ a prefix and $w[i..|w|]$ a suffix of $w$. The reversal of $w$ is denoted by $w^R$, i.e., for $|w|=n$, $w^R = w[n]w[n-1]\dots w[1]$.

We call $A= (Q, \Sigma, \delta, q_0, F)$ a \emph{deterministic finite automaton} (DFA for short) if $Q$ is a finite set of states, $\Sigma$ is a finite input alphabet, $\delta$ is a transition function $Q \times \Sigma \to Q$, $q_0$ is the initial state and $F \subseteq Q$ is the set of final states. The transition function $\delta$ is generalized to words by $\delta(q, w) = \delta(\delta(q, w[1]), w[2..|w|])$ for $w \in \Sigma^*$.
A word $w\in \Sigma^*$ is accepted by $A$ if $\delta(q_0, w) \in F$ and the language accepted by $A$ is defined by $\lang(A) = \{w \in \Sigma^* \mid \delta(q_0, w) \in F\}$.
We extend $\delta$ to sets of states  $Q'\subseteq Q$  or to sets of letters $\Sigma'\subseteq \Sigma$, letting  $\delta(Q',\Sigma')=\{\delta(q',\sigma')\mid  (q',\sigma')\in Q'\times\Sigma'\}$. Similarly, we may  write $\delta(Q',\Sigma')=p$ to define $\delta(q',\sigma')=p$ for each $(q',\sigma')\in Q'\times\Sigma'$. 
The synchronization problem for DFAs (called \textsc{DFA-Sync}) asks for a given DFA $A$ whether there exists a synchronizing word for $A$. A word $w$ is called a \emph{synchronizing word} for a DFA $A$ if it brings all states of the automaton to one single state, i.e., $|\delta(Q, w)| = 1$.

We call $M = (Q, \Sigma, \Gamma, \delta, q_0, \bot, F)$ a \emph{deterministic push-down automaton} (DPDA for short) if $Q$ is a finite set of states; the finite sets $\Sigma$ and $\Gamma$ are the input and stack alphabet, respectively; $\delta$ is a transition function $Q \times \Sigma \times \Gamma \to Q \times \Gamma^*$; 
 $q_0$ is the initial state; $\bot \in \Gamma$ is the stack bottom symbol which is only allowed as the first (lowest) symbol in the stack, i.e., if $\delta(q,a,\gamma)=(q',\gamma')$ and $\gamma'$ contains $\bot$, then $\bot$ only occurs in $\gamma'$ as its prefix and moreover, $\gamma=\bot$; and $F$ is the set of final states. 
We will only consider real-time push-down automata and forbid $\epsilon$-transitions, as can be seen in the definition.
Notice that the bottom symbol can be removed, but then the computation gets stuck.

Following~\cite{DBLP:journals/jalc/ChistikovMS19}, a \emph{configuration} of $M$ is a tuple $(q, \upsilon) \in Q \times \Gamma^*$.
 For a letter $\sigma \in \Sigma$ and a stack content $\upsilon$ with $|\upsilon| = n$ we write $(q, \upsilon) \mto[\sigma] (q', \upsilon[1..(n-1)]\gamma)$ if $\delta(q, \sigma, \upsilon[n]) = (q', \gamma)$.
This means that the top of the stack $\upsilon$ is the right end of  $\upsilon$.  
We also denote with $\longrightarrow$ the reflexive transitive closure of the union of $\mto[\sigma]$ over all letters in $\Sigma$. The input words on top of $\longrightarrow$ are concatenated accordingly, so that $\longrightarrow\, =\bigcup_{w\in\Sigma^*}\mto[w]$.
The language $\lang(M)$ accepted by a DPDA $M$ is 
$\lang(M) = \{w \in \Sigma^* \mid (q_0, \bot) \mto[w] (q_f, \gamma), q_f \in F\}$.
We call the sequence of configurations $(q, \bot) \mto[w] (q', \gamma)$ the \emph{run} induced by $w$, starting in $q$, and ending in $q'$. We might also call $q'$ the \emph{final state} of the run. 

We will discuss three different concepts of synchronizing DPDAs. For all concepts we demand that a synchronizing word $w \in \Sigma^*$ maps all states, starting with an empty stack, to the same synchronizing state, i.e., for all
$q, q' \in Q \colon (q, \bot) \mto (\overline{q}, \upsilon), (q', \bot) \mto (\overline{q}, \upsilon')$. In other words, for a synchronizing word all runs started on some states in $Q$ end up in the same final state.
 In addition to synchronizing the states of a DPDA we will consider the following two conditions for the stack content: 
 (1) $\upsilon = \upsilon' = \bot$, 
 (2) $\upsilon = \upsilon'$. 
We will call (1) the \emph{empty stack model} and (2) the \emph{same stack model}. In the third case, we do not put any restrictions on the stack content and call this the \emph{arbitrary stack model}.

\noindent
As we are only interested in synchronizing a DPDA, we can neglect the start and final states.

As mentioned above, we will show that synchronizability of DPDAs is undecidable, which is in stark contrast to the situation with DFAs, where this problem is solvable in polynomial time. Hence, it is interesting to discuss deterministic variants of classical sub-classes of context-free languages. In this paper, we focus on one-counter languages and on linear languages and related classes.
 A \emph{deterministic (one) counter automaton} (DCA) is a DPDA where $|\Gamma\backslash\{\bot\}| = 1$. Note that our DCAs can perform zero-tests by checking if the bottom-of-stack symbol is on top of the stack. As we will see that also with this restriction, synchronizability is still undecidable, we further restrict them to the \emph{partially blind} setting~\cite{greibach1978blind}. This means in our formalization  that a transition $\delta(q,\sigma,x)=(q',\gamma)$ either satisfies $\gamma=\epsilon$ for both $x=1$ and $x = \bot$, or $x$ is a prefix of~$\gamma$, i.e., $\gamma=x\gamma'$, and then both $\delta(q,\sigma,1)=(q',1\gamma')$ (for $\Gamma=\{1,\bot\}$) and $\delta(q,\sigma,\bot)=(q',\bot\gamma')$. The situation is even more delicate with one-turn or, more general, finite-turn DPDAs, whose further discussion and formal definition we defer to the specific section below.
 
We are now ready to define a family of synchronization problems, the complexity of which will be our subject of study in the following chapters.
\begin{definition}[\sc Sync-DPDA-Empty]
	\ \\
	Given: DPDA $M = (Q, \Sigma, \Gamma, \delta, \bot)$.\\
	Question: Does there exist a word $w \in \Sigma^*$ that synchronizes $M$ in the empty stack model?
\end{definition}
For the same stack model, we refer to the synchronization problem above as \textsc{Sync-DPDA-Same} and as \textsc{Sync-DPDA-Arb} in the arbitrary stack model. Variants of these problems are defined by replacing the DPDA in the definition above by a DCA, a deterministic partially blind counter automaton (DPBCA), or by adding turn restrictions, in particular, whether the automaton is allowed to make zero or one turns of its stack movement.

\subsection*{Outlook and summary of the paper}

We summarize our results in Table~\ref{tab:DPDA-results}. In short, while already seemingly innocuous extensions of finite automata (with counters or with 1-turn push-downs) result in an undecidable synchronizability problem, some extensions 
do offer some algorithmic synchronizability checks, although nothing efficient. At the end, we show how to apply some of our techniques to synchronizability questions concerning sequential transducers.

\begin{table}[tb]
	\begin{tabular}{lccc}
		class of automata/problem & empty stack model & same stack model & arbitrary stack model\\
		\toprule
		DPDA & undecidable & undecidable & undecidable\\
		1-Turn-Sync-DPDA & undecidable & undecidable & undecidable\\
		0-Turn-Sync-DPDA & \PSPACE-complete & undecidable & \PSPACE-complete\\
		DCA & undecidable & undecidable & undecidable\\
		1-Turn-Sync-DCA & \PSPACE-complete & \PSPACE-complete &\PSPACE-complete\\
		0-Turn-Sync-DCA & \PSPACE-complete & \PSPACE-complete & \PSPACE-complete\\
		DPBCA & decidable & decidable & decidable\\
	\end{tabular}
	\caption{Complexity status of the synchronization problem for different classes of deterministic real-time push-down automata in different stack synchronization modes as well as finite-turn variants of the respective synchronization problem.}
	\label{tab:DPDA-results}
\end{table}

As an auxiliary result for proving undecidability of finding 1-turn synchronizing words for real-time  deterministic push-down automata, we also prove  undecidability of the inclusion and intersection non-emptiness problems for these automata, which could be an interesting result on its own. We also showcase how a  variant of DFA synchronization, called \textsc{DFA-Sync-From-Into-Subset}, is useful to prove membership in \PSPACE. As with this type of problems, these membership proofs are sometimes technical, this could be helpful in similar settings.
To understand this problem, first look at the following one, called \textsc{Global Inclusion Problem for Non-Initial Automata} in~\cite{DBLP:journals/ipl/Rystsov83}:
\begin{definition}[\textsc{DFA-Sync-Into-Subset} (\PSPACE-complete, see Theorem~2.1 in~\cite{DBLP:journals/ipl/Rystsov83})]
	\ \\
	Given: DFA $A = (Q, \Sigma, \delta)$, subset $S\subseteq Q$.\\
	Question: Is there a word $w\in \Sigma^*$ such that $\delta(Q, w) \subseteq S$?
\end{definition}
\begin{definition}[\textsc{DFA-Sync-From-Into-Subset}]
	\ \\
	Given: DFA $A = (Q, \Sigma, \delta)$, subsets $S_0,S_1\subseteq Q$.\\
	Question: Is there a word $w\in \Sigma^*$ such that $\delta(S_0, w) \subseteq S_1$?
\end{definition}

Using the previously mentioned \PSPACE-hardness and adapting a subset construction as in \cite[Theorem 1.22]{San2005},
which then boils down to a reachability problem, one can show:
\begin{proposition}\label{Sync-From-Into-Subset}
\textsc{DFA-Sync-From-Into-Subset} is \PSPACE-complete.
\end{proposition}

\noindent
The reader can find proofs for this and several other assertions in the appendix of this paper.

\section{General DCAs and  DPDAs: When Synchronizability is Really Hard}

	The inclusion problem for deterministic real-time one counter automata that can perform zero-tests is undecidable~\cite{10.1007/978-3-642-22993-0_20,minsky1961recursive}. This result is used to prove undecidability of synchronization in any general setting as the main result of this section.
However, there are special cases of DPDAs and DCAs that have a decidable inclusion problem (see~\cite{DBLP:journals/ieicet/HiguchiWT95} as an example) so that this argument does not apply to these sub-classes. We will have a closer look at some of these sub-classes in the following sections.
\begin{theorem}
	\label{thm:DCA-Sync}
	The problems \textsc{Sync-DCA-Empty}, \textsc{Sync-DCA-Same}, and \textsc{Sync-DCA-Arb} are undecidable.
\end{theorem}
\begin{proof}
	We give a reduction from the undecidable intersection non-emptiness problem for real-time DCAs~\cite{10.1007/978-3-642-22993-0_20}.
	Let $M_1 = (Q_1, \Sigma, \{1, \bot\}, \delta_1, q^1_0, \bot, F_1)$ and $M_2 = (Q_2, \Sigma, \{1, \bot\}, \delta_2, q^2_0, \bot, F_2)$ be two DCAs over the same input alphabet with disjoint state sets. We construct a DCA $M_S = (Q_1 \cup Q_2 \cup \{q^1_{f}, q^2_f, q_s\}, \Sigma \cup \{a, b\}, \{1, \bot\}, \delta, \bot)$, where we neglect start and final states, which is synchronizable in the empty stack model if and only if the DCAs $M_1$ and $M_2$ accept a common word. 
	The same construction also works for the same stack and arbitrary stack models.
	We assume $\{q^1_{f}, q^2_f, q_s\} \cap (Q_1 \cup Q_2) = \emptyset$ and $\{a, b\} \cap \Sigma = \emptyset$.
	For the states in $Q_1$ and~$Q_2$, the transition function $\delta$ agrees with $\delta_1$ and $\delta_2$ for all letters in $\Sigma$.
	In the following, let $i \in \{1, 2\}$. For $q \in Q_i$, we set $\delta(q, a, \bot) = (q^i_0, \bot)$ and $\delta(q, a, 1) = (q^i_f, 1)$. 
	Further, for $q \in Q_i \backslash F_i$ we set $\delta(q, b, 1) = (q^i_f, 1)$ and $\delta(q, b, \bot) = (q^i_f, \bot1)$. For $q \in F_i$, we set $\delta(q, b, 1) = (q_s, 1)$ and $\delta(q, b, \bot) = (q_s, \bot)$.
	For $q^i_f$, we set $\delta(q^i_f, a, \bot) = (q^i_0, \bot)$ with all other transitions we stay in $q^i_f$ and increase the counter. Hence, the state $q^i_f$ can only be left with an empty counter and this is only the case if no letter other than $a$ has been read before.
	For the state $q_s$, we set $\delta(q_s, \Sigma\cup\{a, b\}, \bot) = (q_s, \bot)$, and $\delta(q_s, \Sigma \cup \{a, b\}, 1) = (q_s, \epsilon)$.
	
	First, assume there is a word $w \in \lang(M_1) \cap \lang(M_2)$. Then, the word $awb$ synchronizes all states of the DCA $M_S$ into the state $q_s$. Let $l_1, l_2$ be stack contents such that $(q^1_0, \bot) \mto[awb] (q_s, l_1)$ and $(q^2_0, \bot) \mto[awb] (q_s, l_2)$. Let $l = \max(|l_1|, |l_2|)$. Then $awbb^l$ synchronizes $M_S$ in the empty stack model.
	
	For the other direction assume there exists a word $w \in (\Sigma \cup \{a, b\})^*$ that synchronizes $M_S$ in the empty stack model. The states $q^1_f$ and $q^2_f$ forces $w[1] = a$ since otherwise these states cannot be left. Since the state $q_s$ has no outgoing transition, it must be our synchronizing state. In order to reach it, $w$ must contain at least one letter $b$. Let $m \in [|w|]$ be an index such that $w[m] = b$ and for $j < b$, $w[j] \neq b$. With a letter $b$ we move from all final states to the state $q_s$ and from all non-final states of $M_1$ and $M_2$ we go to a state $q^i_f$ and increase the counter. As we cannot leave the states $q^i_f$ if we reach them once with a non-empty counter, reading a $b$ from a non-final state causes the automaton to reach a configuration from which we no longer can synchronize the automaton. Hence, we know that after reading $w[1..m]$ all active states are in the set $F_1 \cup F_2 \cup \{q_s\}$. Let $\ell \in [|w|]$ be an index with $\ell < m$ with $w[\ell] = a$ such that for $\ell < i <m$, $w[i] \neq a$. Then $w[\ell+1\, ..\,m-1]$ is a word which is accepted by both DCAs $M_1$ and $M_2$.
	
	It is easy to see that clearing the stacks in state $q_s$ is not crucial and hence the reduction also works for the same stack and arbitrary stack models.
\end{proof}
\begin{corollary}
	The problems \textsc{Sync-DPDA-Empty}, \textsc{Sync-DPDA-Same}, and \textsc{Sync-DPDA-Arb} are undecidable.\qed
\end{corollary}

How can we overcome the problem that, even for deterministic one-counter languages, the synchronizability problem is undecidable?
One of the famous further restrictions are (partially) blind counters, to which we turn our attention next.

\section{Partially Blind Deterministic Counter Automata}


The blind and partially blind variations of counter automata have been introduced by Greibach in~\cite{greibach1978blind}.
She already noticed that the emptiness problem for such automata (even with multiple counters) is decidable should the reachability problem for 
vector addition systems, also known as Petri nets, be decidable, which has been proven some years later \cite{Mayr1981reachability,Kosaraju1982reachability}; its non-elementary complexity has only been recently fully understood~\cite{czerwinskietal2019reachability}. Although we will stick to the models introduced so far in the following statements and proofs, we want to make explicit that our decidability results also hold for deterministic multi-counter automata. But as we focus on discussing families of automata describing languages between regular and context-free, we refrain from giving further details here.

Because partially blind counters can simulate blind counters, our results hold for blind counters as well, but we make them explicit only for the partially blind case. One formal reason is that we want to preserve our stack model, while it becomes awkward to formalize blind counters in this stack model (see Appendix).

Recall that a partially blind counter automaton will get blocked when its counter gets below zero. 
The blindness refers to the fact that such a machine can never explicitly test its counter for zero.
This translates into our formalization by requiring that a transition $\delta(q,\sigma,x)=(q',\gamma)$ either means that $x\neq\bot$ or $x$ is a prefix of $\gamma$, i.e., $\gamma=x\gamma'$, and then both $\delta(q,\sigma,1)=(q',1\gamma')$ (for $\Gamma=\{1,\bot\}$) and $\delta(q,\sigma,\bot)=(q',\bot\gamma')$, i.e., 
the processing will somehow perform the same action, irrespectively of the stack contents, whenever this is possible; however, the machine will 
stop if it is
trying to pop the bottom-of-stack symbol. As a specialty, such automata accept when having arrived in a final state together with having zero in its counter. We call a deterministic partially blind (one-)counter automaton a DPBCA.

\begin{theorem}
	\label{thm:DPBCA-Sync}
	The problems \textsc{Sync-DPBCA-Empty}, \textsc{Sync-DPBCA-Same}, and \textsc{Sync-DPBCA-Arb} are decidable.
\end{theorem}

We are not specifying the complexity here, but only mention that we are using, in the end, the reachability problem for Petri nets, which 
is known to be decidable, but only with a non-elementary complexity; see \cite{Mayr1981reachability,Kosaraju1982reachability,czerwinskietal2019reachability}. However, we leave it as an open question if the synchronization complexity of DPBCAs is non-elementary. When looking into this question more in details, the number of states of the counter automaton could be a useful parameter to be discussed, as it influences the number of counters of the partially blind multi-counter automaton that we construct in our proof in order to show the claimed decidability result.

\begin{proof}  Let $M = (Q, \Sigma, \{1, \bot\}, \delta, q_0, \bot, F)$ be some DPBCA. 
Let us first describe the case  \textsc{Sync-DPBCA-Empty}.
Here, we can first produce the multi-counter $|Q|$-fold product automaton $M^{|Q|}$ from $M$ that starts, assuming $Q=\{q_0,\dots,q_{|Q|-1}\}$, in the 
state $(q_0,\dots,q_{|Q|-1})$. Notice that  $M^{|Q|}$ has $|Q|^{|Q|}$ many states and operates $|Q|$ many counters. 
We could take as the set of final states $F^{|Q|}=\{(q,\dots,q)\mid q\in Q\}$. 
This mimicks state synchronization of $M$: any word that synchronizes all states of $M$ will drive $M^{|Q|}$ into $F$.  
As mentioned above, partially blind multi-counter automata accept with final states and empty stacks, so that $M$ is synchronizable in the empty stack model if and only if $M^{|Q|}$ accepts any word.

For the arbitrary stack model, 
we have to count down (removing~$1$ from any of the the stacks sequentially) until the bottom-of-stack symbol appears on all stacks on top (at the same time),
leading to the variant $M_{\textsc{Arb}}^{|Q|}$. These are moves without reading the input (or reading arbitrary symbols at the end, this way only prolonging a possibly synchronizing word), but this does not matter, as the emptiness problem is decidable for  partially blind nondeterministic multi-counter automata. It should be clear that $M^{|Q|}_{\textsc{Arb}}$ accepts any word if and only if $M$ is synchronizable. 
 For the case \textsc{Sync-DPBCA-Same}, the counting down at the end
should be performed  in parallel for all counters instead. 
\end{proof}


\section{Finite-Turn DPDAs}
Finite-turn PDAs are introduced in~\cite{ginsburg1966finite}. From the formal language side, it is known that one-turn PDAs characterize the rather familiar family of linear context-free languages, usually defined via grammars. In our setting, the automata view is more interesting. 
We adopt the definition in~\cite{valiant1973decision}.
For a DPDA $M$ an \emph{upstroke} of $M$ is a sequence of configurations induced by an input word $w$ such that no transition decreases the stack-height. Accordingly a \emph{downstroke} of $M$ is a sequence of configurations in which no transition increases the stack-height. A stroke is either an upstroke or downstroke. Note that exchanging the top symbol of the stack is allowed in both an up- and downstroke.
A DPDA $M$ is an $n$-turn DPDA if for all $w \in \lang(M)$ the sequence of configurations induced by $w$ can be split into at most $n+1$ strokes. Especially, for 1-turn DPDAs each sequence of configurations induced by an accepting word consists of one upstroke followed by a most one downstroke.
There are two subtleties when translating this concept to synchronization: (a) there is no initial state so that there is no way to associate a stroke counter to a state, and (b) there is no language of accepted words that restricts the set of words on which
the number of strokes should be limited. 
We therefore generalize the concept of finite-turn DPDAs to finite-turn synchronization for DPDAs in the following way, which offers quite an interesting complexity landscape.
\begin{definition} \textsc{$n$-Turn-Sync-DPDA-Empty}
	\ \\
	Given: DPDA $M = (Q, \Sigma, \Gamma, \delta, q_0, \bot, F)$.\\
	Question: Is there a synchronizing word $w\in \Sigma^*$ in the empty stack model, such that 
	for all states $q \in Q$, the sequence of configurations $(q, \bot) \mto[w] (\overline{q}, \bot)$ consists of at most $n+1$ strokes?
\end{definition}
We call such a synchronizing word $w$ an \emph{$n$-turn synchronizing word} for $M$.
We define \textsc{$n$-Turn-Sync-DPDA-Same} and \textsc{$n$-Turn-Sync-DPDA-Arb} accordingly for the same stack and arbitrary stack models. Further, we extend the problem definition to real-time DCAs.

Motivated by the proof of Theorem~\ref{thm:DCA-Sync}, we are first reviewing the status of the inclusion problem for 1-turn DPDAs in the literature.

\begin{remark}
	The inclusion problem for 1-turn DPDAs with $\epsilon$-transitions is undecidable~\cite{DBLP:journals/tcs/Friedman76,valiant1973decision}. The intersection non-emptiness problem for real-time 1-turn non-deterministic push-down automata is also undecidable~\cite{DBLP:journals/tcs/Kim11a}. 
	The decidability of the inclusion and intersection non-emptiness problems for \emph{real-time} 1-turn \emph{deterministic} push-down automata have not been  settled in the literature; we will do so below by proving undecidability for both problems.
\end{remark}
We will present a reduction from the undecidable \textsc{Post Correspondence Problem} (\textsc{PCP} for short)~\cite{Pos46} to the intersection non-emptiness for real-time 1-turn DPDAs which also implies undecidability of the inclusion problem for this class since it is closed under complement.
\begin{definition}[\textsc{PCP}]
	\ \\
	Given: Two lists of input words over $\{0, 1\}$: $A = (a_1, a_2, \dots, a_n)$, $B=(b_1, b_2, \dots, b_n)$.\\
	Question: Is there a sequence of indices $i_1, i_2, \dots, i_k$ with $i_j \in [n]$ for $1 \leq j \leq k$ such that $a_{i_1}a_{i_2} \dots a_{i_k} = b_{i_1}b_{i_2}\dots b_{i_k}$?	
\end{definition}

Observe that already Post stated this problem over binary alphabets. Much later, Matyasevich and Sénizergues~\cite{MatSen2005} showed that indeed lists of length seven are sufficient for undecidability. This was recently lowered to lists of length five by Neary~\cite{DBLP:conf/stacs/Neary15}.

\begin{theorem}
	\label{thm:1-Turn-Intersect}
	Let $M_1$ and $M_2$ be two real-time 1-turn DPDAs. Then the following problems are undecidable: Is $\lang(M_1) \cap \lang(M_2) = \emptyset$? Is $\lang(M_1) \subseteq \lang(M_2)$?
\end{theorem}
\begin{proof}
	Let $A = (a_1, a_2, \dots, a_n)$, $B=(b_1, b_2, \dots, b_n)$ be an instance of PCP. We construct from $A$ a real-time 1-turn DPDA $M_A = (Q, \{0, 1, \#, \$\} \cup [\overline{n}], \{0, 1, \bot\}, \delta, q_0, \bot, \{q_f\})$ where $[\overline{n}] = \{\overline{1}, \overline{2}, \dots, \overline{n}\}$ are marked numbers from 1 to $n$. The set $Q$ contains the start state $q_0$, the states $\overline{q_0}$, $q_{\text{check}}$ and $q_{\text{fail}}$, and the single final state $q_f$. The rest of $Q$ is a partition into state sets $Q_1, Q_2, \dots, Q_n$ such that the (deterministic partial) sub-automaton induced by $Q_i$ reads the string $a_i\#b_i$ and thereby pushes each symbol of $a_i$ on the stack, whereas symbols of the string $b_i$ leave the stack content unchanged. The sub-automaton induced by $Q_i$ is embedded into $M_A$, and thereby completed, by taking the state after reading $a_i\#b_i$ to $\overline{q_0}$ with the symbol~$\#$. 
	We go from $q_0$ and $\overline{q_0}$ to the initial state of the sub-automaton induced by $Q_i$ by the letter $\overline{i}$. With the letter $\$$ the state $\overline{q_0}$ maps to $q_{\text{check}}$. Here, if the input symbol equals the symbol on top of the stack, we pop the stack and stay in $q_{\text{check}}$ until we reach the bottom symbol $\bot$, in which case an additional letter~$\$$ brings us to the final state $q_f$. If the input symbol does not equal the symbol on top of the state we go to $q_{\text{fail}}$ which is a trap state for all letters. Every other not yet defined transition on $Q$ maps to the state $q_{\text{fail}}$. 
	If not stated otherwise, every transition leaves the stack content unchanged.
	
	For the list $B$ we construct a real-time 1-turn DPDA $M_B$ in a similar way except that here we push the strings $b_i$ on the stack and symbols of strings $a_i$ leave the stack unchanged.
	
	
The languages accepted by $M_A$ and $M_B$ are the following ones:
$$ \lang(M_A) =\{\,\overline{i_1}a_{i_1}\#b_{i_1}\#\overline{i_2}a_{i_2}\#b_{i_2}\#\cdots \overline{i_m}a_{i_m}\#b_{i_m}\# \$ a_{i_m}^R\cdots a_{i_1}^R\$\mid m\geq 1, i_j\in [n], j \leq m\}$$
$$ \lang(M_B) =\{\,\overline{i_1}a_{i_1}\#b_{i_1}\#\overline{i_2}a_{i_2}\#b_{i_2}\#\cdots \overline{i_m}a_{i_m}\#b_{i_m}\# \$ b_{i_m}^R\cdots b_{i_1}^R\$\mid m\geq 1, i_j\in [n], j \leq m\}$$
Obviously, the given PCP has a solution if and only if 
$\lang(M_A) \cap \lang(M_B)\neq \emptyset$. 

By complementing the set of final states, from $M_A$ one would arrive at a real-time 1-turn DPDA $M_A'$ such that $\lang(M_A')$ is the complement of $\lang(M_A)$, so that  the given PCP has no solution if and only if 
$\lang(M_A')\supseteq \lang(M_B)$.
%
\end{proof}
We will now adapt the presented construction to show that the synchronization problem for real-time one-turn DPDAs is undecidable in all three synchronization models. 

\begin{theorem}
	\label{thm:1-turn-sync}
	\textsc{1-Turn-Sync-DPDA-Empty}, \textsc{1-Turn-Sync-DPDA-Same}, and \textsc{1-Turn-Sync-DPDA-Arb} are undecidable.
\end{theorem}
\begin{proof}
	Let $M_A$ and $M_B$ be the real-time 1-turn DPDAs from the proof in Theorem~\ref{thm:1-Turn-Intersect}. 
	We take these machines as inputs for the construction in the proof of Theorem~\ref{thm:DCA-Sync} to obtain the DPDA $M$. Therefore, observe that the construction also works if the input machines are general DPDAs and not only DCAs. Further, observe that for $M_A$ and $M_B$, if for both machines the only active state is the final state, then the stack of all runs is empty and hence the stack need not be altered in the synchronizing state and all transitions here can act as the identity and leave the stack unchanged. 
	
	If $w$ is a word in $\lang(M_A) \cap \lang(M_B)$, then $awb$ synchronizes $M$ in the empty, same, and arbitrary stack models; further $awb$ is a 1-turn synchronizing word for $M$.
	Conversely, if~$w$ is a 1-turn synchronizing word for $M$, then $w$ must be of the form $avb$ or $auavb$ where $v \in \lang(M_A) \cap \lang(M_B)$ and $u$ is a word that does not change the stack, as otherwise $M$ could either not be synchronized or the 1-turn condition is violated.
	To be more precise, $a$ must be the first letter of $w$ as otherwise we get stuck in $q_f^1$ and $q_f^2$. The letter $a$ resets the machines $M_A$ and $M_B$ to their initial state and can only be read when the stack is empty as otherwise the machine gets stuck. In order to reach final states, both machines $M_A$ and $M_B$ must increase and decrease the stack by reading some word $v$, but as soon as we increased the stack once, we are not allowed to reset the machine anymore due to the 1-turn condition. Hence, letters $a$ can only be read while the stack has not been changed yet.
		Note that for all three stack models, the construction enforces that any 1-turn synchronizing word brings $M$ into a configuration where the stack is empty.
\end{proof}

When considering automata as language accepting devices, there is no good use of 0-turn PDAs, as they cannot exploit their stack. This becomes different if synchronization requires to end in the same configuration, which means that in particular the stack contents are identical.

\begin{theorem}
	\label{thm:0-turn-Same-Sync}
	The problem \textsc{0-Turn-Sync-DPDA-Same} is undecidable.
\end{theorem}
\begin{proofsketch}
	The proof is by a straight-forward adaption of the previously presented constructions by getting rid of the check phase; instead, check with the same stack condition that the two words of the PCP coincide. As we never pop the stack, a 0-turn DPDA will be sufficient in the construction.
\end{proofsketch}

\noindent
The picture changes again for other 0-turn stack models, but remains intractable.

\begin{theorem}
	\label{thm:0-turn-Arb-DPDA}
	The problems \textsc{0-Turn-Sync-DCA-Empty}, \textsc{0-Turn-Sync-DCA-Same} and \textsc{0-Turn-Sync-DCA-Arb} are \PSPACE-hard.
\end{theorem}
\begin{proof}
	We give a reduction from  \textsc{DFA-Sync-Into-Subset}. Let $A = (Q, \Sigma, \delta)$ be a DFA with $S \subseteq Q$. We construct from $A$ a DCA $M = (Q \cup \{q_{\text{stall}}, q_{\text{sync}}\}, \Sigma \cup \{a\}, \{N, \bot\}, \delta', \bot)$ where all unions are disjoint.
	For $q\in Q$, $\sigma\in\Sigma$  and $\gamma\in \{N, \bot\}$, set $\delta'(q,\sigma,\gamma)=(\delta(q,\sigma),\gamma)$.
	For the letter $a$, we set for states $q \in S$, $\delta'(q, a, \bot) = (q_{\text{stall}}, \bot)$ and for states $q\in Q\backslash S$, we set $\delta'(q, a, \bot) = (q_{\text{stall}}, N)$. For $q_{\text{stall}}$ we set $\delta'(q_{\text{stall}}, a, \bot) = (q_{\text{sync}}, \bot)$ and $\delta'(q_{\text{stall}}, a, N) = (q_{\text{stall}}, N)$. All transitions not yet defined act as the identity and leave the stack unchanged.
	
	First, assume there exists a word $w \in \Sigma^*$ that synchronizes~$Q$ into~$S$ in the DFA~$A$. Then clearly $waa$ synchronizes~$M$ in the arbitrary stack model. 
	Now, assume there exists a word $w \in (\Sigma \cup \{a\})^*$ that synchronizes~$M$ in the arbitrary stack model. Then, $w$ must contain at least two occurrences of~$a$ to bring all states into the sink state $q_{\text{sync}}$. In order to reach $q_{\text{sync}}$ the states in~$Q$ need to pass through the state $q_{\text{stall}}$ but by doing so, it is noted on the stack if an active state transitions from $Q\backslash S$ into $q_{\text{stall}}$ and only the active states coming from~$S$ are allowed to pass on to $q_{\text{sync}}$. In $q_{\text{stall}}$ the stack content cannot be changed and hence the prefix of~$w$ up to the first occurrence of the letter~$a$ must have already synchronized~$Q$ into~$S$ in the DFA~$A$. Note that $M$ can only be synchronized by a word that leaves all stacks empty, hence the result follows for all three stack models.
\end{proof}
\begin{corollary}
	The problems 
\textsc{0-Turn-Sync-DPDA-Empty} and  \textsc{0-Turn-Sync-DPDA-Arb} are \PSPACE-hard.
\end{corollary}
\begin{proof}
	The claim follows from Theorem~\ref{thm:0-turn-Arb-DPDA} by inclusion of automata classes.
\end{proof}
\begin{theorem}
	\label{thm:DPDA-0-turn-PSPACE-Membership}
$\textsc{0-Turn-Sync-DPDA-Empty}, \textsc{0-Turn-Sync-DPDA-Arb}\in \PSPACE$.
\end{theorem}
\begin{proofsketch}
	In the empty stack model, the 0-turn condition forbids us to write anything on the stack at all.
Hence, consider a 0-turn DPDA  with state set $Q$ as a DFA $A$, with one additional state $q_{\text{trap}}$ which is entered whenever the DPDA wants to push or pop anything. For each $q\in Q$,
run the \PSPACE-algorithm for the \textsc{Sync-From-Into-Subset} instance $(A,Q,\{q\})$. 
For the arbitrary stack model, due to the 0-turn condition, for each run only the symbol on top of the stack must be ever stored.
Consider a 0-turn  DPDA  with state set $Q$ and stack alphabet $\Gamma$ as a DFA $A$ with state set $Q\times\Gamma$ and one additional state $q_{\text{trap}}$ which is entered whenever the DPDA wants to pop anything. For each $q\in Q$,
run the \PSPACE-algorithm for the \textsc{Sync-From-Into-Subset} instance $(A,Q\times\{\bot\},\{q\}\times\Gamma)$. 
%
\end{proofsketch}

\begin{theorem}
	\label{thm:DCA-1-turn-PSPACE-Membership}
	 \textsc{1-Turn-Sync-DCA-Empty}, \textsc{1-Turn-Sync-DCA-Same}, and \textsc{1-Turn-Sync-DCA-Arb} in \PSPACE.
\end{theorem}
Notice that \PSPACE-hardness is inherited from corresponding results for visibly counter automata, as obtained in \cite{vpda-crossref}.
\begin{proof}Let $M = (Q , \Sigma, \Gamma, \delta, \bot)$ be a DCA. As we are looking into 1-turn behavior, any computation that we are interested would split into two phases: in the first upstroke phase, the counter is incremented or stays constant, while in the second downstroke phase, the counter is decremented or stays constant.
In particular, because the counter is 1-turn, after the first counter increment, any zero-test will always return false, while in the downstroke phase, when zero-tests return true, then all future computations cannot decrement the counter any further, so that at the end, the counter will also contain zero.
We are formalizing this intuition to create a machine that has an awareness about its phase stored in its states and that is behaving very similar.
In the rest of the proof, a \emph{spread-out variant} of a word $a_1a_2\cdots a_n$ of length $n$, with symbols $a_i$ from $\Sigma$, is any word in $a_1\Sigma a_2\Sigma \cdots a_n\Sigma$ of length $2n$.

We will now  construct from $M$ and $q\in Q$ a deterministic  1-turn counter automaton $M_q$ that accepts precisely all spread-out variants of words that $M$ would accept when starting in state $q$ and finishing its 1-turn computation (in any state) with the empty stack, but that keeps track of a basic property of managing the counter in so-called stages.
 $M_q$ has the state set $Q\times\{1,2,3,4\}\times\{0,1\}$, $(q,1,0)$ as its initial state, and as its set of final states, take $Q\times\{1,4\}\times\{0\}$. The transitions of $\delta_q$ 
can be defined with the following semantics in mind (details are given in the appendix): 
(a) the last bit always alternates, (b) the spread-out is used to enable a \emph{deterministic} work and to make sure that the simulated machine has counter content zero if the simulating automaton~$M_q$ is in one of the states from $Q\times\{1,4\}\times\{0\}$, 
(c)  $M_q$ changes from  $Q\times\{1\}\times\{0,1\}$ to  $Q\times\{2\}\times\{0,1\}$ if the counter is no longer zero, so that the simulated machine has ``properly'' entered the upstroke phase, (d) $M_q$ changes from  $Q\times\{2\}\times\{0,1\}$ to  $Q\times\{3\}\times\{0,1\}$ if the simulated machine made its first pop, i.e., it ``properly'' entered the downstroke phase,
(e) $M_q$ changes from  $Q\times\{3\}\times\{0,1\}$ to  $Q\times\{4\}\times\{0,1\}$ if the counter has become zero again.

Now, we build the $|Q|$-fold product automaton $M^{|Q|}_{\text{Empty}}$ from all automata $M_q$ with the start state $((q_1,1,0),(q_2,1,0),\dots,(q_{|Q|},1,0))$, assuming $Q=\{q_1,\dots,q_{|Q|}\}$. This means that $M^{|Q|}_{\text{Empty}}$ has $(8|Q|)^{|Q|}$ many states and $|Q|$ many counters, each of which makes at most one turn. Now observe that a word $w$ synchronizes $M$ with empty stack, say, in state $p$, if and only if any spread-out variant of $w$ drives $M^{|Q|}_{\text{Empty}}$ into a state $((p,i_1,0),(p,i_2,0),\dots,(p,i_{|Q|},0))$ for some $i_j\in\{1,4\}$ for all $1\leq j\leq |Q|$. Now, define $\{((p,i_1,0),(p,i_2,0),\dots,(p,i_{|Q|},0))\mid p\in Q, i_j\in\{1,4\} \text{ for } 1\leq j\leq |Q|\}$ as final states of $M^{|Q|}_{\text{Empty}}$. We see that $M$ is synchronizable with empty stack if and only if $M^{|Q|}_{\text{Empty}}$ accepts any word. As Gurari and Ibarra have shown in \cite[Lemma~2]{DBLP:journals/jcss/GurariI81}, $M^{|Q|}_{\text{Empty}}$ accepts any word if and only if it accepts any word up to length $(|Q|(8|Q|)^{|Q|}|\Sigma|)^{O(|Q|)}\leq (|Q|(8|Q||\Sigma|)^{O(|Q|)^2})$,\footnote{In the cited lemma, the number of steps of the checking machine is upper-bounded by $(ms)^{O(m)}$, where
$m$ is the number of 1-turn counters and $s$ is the number of transitions of the machine.} within the same time bounds. Now, testing all these words for membership basically needs  two counters that are able to capture numbers of size $(|Q|(8|Q||\Sigma|)^{O(|Q|)^2})$, which means we need polynomial space in $|Q|$ and $|\Sigma|$ to check if $M$ is synchronizable with empty stack.

By considering $M^{|Q|}_{\text{Empty}}$ with final states $((p,i_1,0),(p,i_2,0),\dots,(p,i_{|Q|},0))$ for arbitrary $p\in Q$ and $i_j\in\{1,2,3,4\}$, this way defining an automaton $M^{|Q|}_{\text{Arb}}$, we can check if $M$ is synchronizable in the arbitrary stack model also in polynomial space with the same argument.
From $M^{|Q|}_{\text{Empty}}$, we can also construct a nondeterministic 1-turn $|Q|$-counter machine $M^{|Q|}_{\text{Same}}$ by adding a nondeterministic move from  $((p,i_1,0),(p,i_2,0),\dots,(p,i_{|Q|},0))$ for arbitrary $p\in Q$ and $i_j\in\{1,2,3,4\}$ upon reading some arbitrary but fixed $\sigma_\text{sync}\in\Sigma$ to a special state $q_\text{dec}$  in which state we loop upon reading  $\sigma_\text{sync}\in\Sigma$, decrementing all counters at the same time; finally, there is the possibility to move to $q_f$ 
(that is the only finaly state now) upon reading  $\sigma_\text{sync}\in\Sigma$ if all counters are empty. Notice that also this automaton $M^{|Q|}_{\text{Same}}$ has $(\mathcal{O}(|Q|))^{|Q|}|\Sigma|$ many transitions, so that with using \cite[Lemma~2]{DBLP:journals/jcss/GurariI81}, we can again conclude that synchronizability with same stack can be checked in polynomial space for $M$.
\end{proof}

\section{Sequential Transducers}
We will now introduce a new concept of synchronization of sequential transducers.\footnote{The definitions in the literature are not very clear for finite automata with outputs. We follow here the name used by Berstel in~\cite{DBLP:books/lib/Berstel79}; Ginsburg~\cite{Ginsburg66} called Berstel's sequential transducers \emph{generalized machines}, but used the term \emph{sequential transducer} for the nondeterministic counterpart.
} 
We call $T = (Q, \Sigma, \Gamma, q_0, \delta, F)$ a \emph{sequential transducer} (ST for short) if $Q$ is a finite set of states, $\Sigma$ is a finite input alphabet, $\Gamma$ is a finite output alphabet, $q_0$ is the start state, $\delta \colon Q \times \Sigma \to Q \times \Gamma^*$ is a total transition function, and $F$ is a set of final states. 
We generalize $\delta$ from input letters to words by concatenating the produced outputs, i.e., for $q, q', q'' \in Q$, $\sigma_1, \sigma_2 \in \Sigma$, $\gamma_1, \gamma_2 \in \Gamma^*$ and transitions $\delta(q, \sigma_1) = (q', \gamma_1), \delta(q', \sigma_2) = (q'', \gamma_2)$ we define $\delta(q, \sigma_1\sigma_2) = (q'', \gamma_1\gamma_2)$.
We say that a word $w$ \emph{trace-synchronizes} a sequential transducer $T$ if for all states $p, q \in Q$ it holds that $\delta(p, w) = \delta(q, w)$. Intuitively, $w$ brings all states of $T$ to the same state and produces the same output on all states.
Again, we might neglect start and final states.

\begin{definition}[\sc Trace-Sync-Transducer]
	\ \\
	Given: Sequential transducer $T = (Q, \Sigma, \Gamma, \delta)$.\\
	Question: Does there exists a word $w \in \Sigma^*$ that trace-synchronizes $T$?
\end{definition}
\begin{theorem}
	The problem \textsc{Trace-Sync-Transducer} is undecidable.
\end{theorem}
\begin{proof}
	We adapt the construction of $M$ in Theorem~\ref{thm:0-turn-Same-Sync} to obtain a sequential transducer $T$ in the following way:
	Each time, we push a letter to the stack in $Q^A_i$ or $Q^B_i$, instead we now output this letter. Whenever we leave the stack unchanged, we now simply do not produce any output. For the letter $a$, we output the special letter $r$ on all states in order to indicate that the machine has been reset. For the state $q^A_{\text{fail}}$ we output the special letter $A$ and for $q^B_{\text{fail}}$ the special letter $B$, for all input symbols expect $a$.
	
	We make the following observations for any potential trace-synchronizing word $w$ for $T$: (1) $w$ needs to start with $a$. (2) $w$ needs to contain the sub-word $\#\#$. (3) The sub-word between the first occurrence of $\#\#$ and the respective last occurrence of $a$ before $\#\#$ describes a solution of the PCP instance. (4) If the PCP instance has a solution, one can construct a trace-synchronizing word for $T$ from that solution.
\end{proof}

\section{Prospects}

It would be interesting to look into the synchronization problem for  further automata models. In view of the undecidability results that we obtained in this paper, a special focus should be to look into deterministic automata classes with a known decidable inclusion problem, as otherwise it should be possible to adapt our undecidability proofs for synchronizability to these automata models. To make this research direction more clear: (a) There are quite efficient algorithms for the inclusion problem for so-called \emph{very simple deterministic pushdown automata}, see~\cite{WakTom93}; (b) a proper super-class of these languages are so-called \emph{NTS languages} that also have a deterministic automaton characterization\footnote{This is rather implicit in the literature, which is one of the reasons why we do not present more details here; one would have to first define the automaton model properly.} but their inclusion problem is undecidable, see~\cite{DBLP:journals/jcss/Senizergues85,DBLP:journals/jcss/BoassonS85}.
The overall aim of this research would be to find the borderline between decidable and undecidable synchronizability and, moreover, within the decidable part, to determine the complexity of this problem. A step in this research direction has been made in direction of (sub-classes of) visibly deterministic pushdown automata in~\cite{vpda-crossref}. Interestingly enough, that research line also revealed some cases where synchronizability can be decided in polynomial time, quite in contrast to the situation found in the present study.

Another approach is to look into variants of synchronization problems for DPDAs, such as restricting the length of a potential synchronizing word. It follows from the \NP-hardness of this problem for DFAs~\cite{Rys80,DBLP:journals/siamcomp/Eppstein90} and the polynomial-time solvability of the membership problem for DPDAs that for unary encoded length bounds this problem is \NP-complete for DPDAs as well, and contained in \EXP\ for binary encoded length bounds. The precise complexity of this problem for binary encoded length bounds will be a topic of future research. 

\bibliography{bib}

\newpage

\section{Appendix: Omitted Proof Details}

\begin{proof}[Proof of Proposition~\ref{Sync-From-Into-Subset}]
	The \PSPACE-hardness can easily be inherited from the \PSPACE-complete problem \textsc{DFA-Sync-Into-Subset} by setting $S_0 = Q$ and $S_1 = S$ on a \textsc{DFA-Sync-Into-Subset} instance consisting of a DFA $A=(Q, \Sigma, \delta)$, and subset $S \subseteq Q$.
	
	For membership in \PSPACE, we adapt the proof in \cite[Theorem 1.22]{San2005} and give a non-deterministic algorithm for solving \textsc{DFA-Sync-From-Into-Subset} in \NPSPACE. The claim then follows by Savitch's famous theorem showing $\PSPACE=\NPSPACE$~\cite{DBLP:journals/jcss/Savitch70}.
	Let $A=(Q, \Sigma, \delta)$, $S_0, S_1\subseteq Q$ be an instance of \textsc{DFA-Sync-From-Into-Subset}. We start on the set $Q_0 = S_0$ and guess a symbol $\sigma_0 \in \Sigma$. Applying $\sigma_0$ to all states in $Q_0$ gives us the set $Q_1 = \delta(Q_0, \sigma_0)$. We continue in this way by guessing a symbol $\sigma_i\in \Sigma$ and applying it to the current set $Q_i$ to obtain $Q_{i+1} = \delta(Q_i, \sigma_i)$ until we reach a set $Q_n$ for which $Q_n \subseteq S_1$ holds and the algorithm terminates with a positive answer. At each step, we only need to keep the current set $Q_i$ in memory such that the algorithm can be performed in polynomial space. 
\end{proof}

\begin{proof}[Proof of Theorem~\ref{thm:0-turn-Same-Sync}]
	The following construction is an easy adaption of the previously presented constructions by getting rid of the check phase and instead check that the two words of the PCP coincide with the same stack condition. As we de not pop the stack a 0-turn DPDA will be sufficient in the construction.
	
	Let $A = (a_1, a_2, \dots, a_n)$, $B=(b_1, b_2, \dots, b_n)$ be an instance of PCP. We construct from~$A$ a real-time 0-turn partial DPDA $M_A = (Q^A, \{0, 1, \#\} \cup [\overline{n}], \{0, 1, \bot\}, \delta^A, \bot)$ where $[\overline{n}] = \{\overline{1}, \overline{2}, \dots, \overline{n}\}$ are marked numbers from 1 to $n$. The set $Q^A$ contains the states $q^A_0$ and $\underline{q^A_0}$. The rest of $Q^A$ is a partition into states $Q^A_1, Q^A_2, \dots, Q^A_n$ such that the (deterministic partial) sub-automaton induced by $Q^A_i$ reads the string $a_i\#b_i$ and thereby pushes each symbol of $a_i$ on the stack whereas symbols of $b_i$ leave the stack unchanged. The sub-automaton induced by $Q^A_i$ is embedded into $M_A$, by taking the state after reading $a_i\#b_i$ to $\underline{q^A_0}$ with the symbol~$\#$. 
	We go from $q^A_0$ and $\underline{q^A_0}$ to the initial state of the sub-automaton induced by $Q^A_i$ by the letter~$\overline{i}$.
	If not stated otherwise, every transition leaves the stack content unchanged.
	
	For the list $B$, we construct a real-time 0-turn partial DPDA $M_B = (Q^B, \{0, 1, \#\} \cup [\overline{n}], \{0, 1, \bot\}, \delta^B, \bot)$ in a similar way, except that here we push the strings $b_i$ on the stack and symbols of strings $a_i$ leave the stack unchanged.
	
	We combine the two machines $M_A$ and $M_B$ into a real-time 0-turn complete DPDA $M = (Q^A \cup Q^B \cup \{q^A_{\text{fail}}, q^B_{\text{fail}}, q_{\text{sync}}\}, \{0, 1, \#, a\}\cup [\overline{n}], \{0, 1, \bot\}, \delta, \bot)$ where all unions are assumed to be disjoint. The transition function $\delta$ agrees with $\delta^A$ and $\delta^B$ on all states in $Q^A \cup Q^B$ and letters in $\{0, 1, \#\}\cup[\overline{n}]$. For the letter $a$ we set for all states $q^j\in Q^j$ with $j \in \{A, B\}$, $\delta(q^j, a, \bot) = (q^j_0, \bot)$ and for $\gamma \neq \bot$, $\delta(q^j, a, \gamma) = (q^j_{\text{fail}}, \gamma)$. 
	For $\underline{q^j_0}$ and $\gamma \in \{0, 1\}$ we set $\delta(\underline{q^j_0}, \#, \gamma) = (q_{\text{sync}}, \gamma)$. For $q^j_{\text{fail}}$ we set $\delta(q^j_{\text{fail}}, a, \bot) = (q^j_0, \bot)$ and for all other transitions we stay in $q^j_{\text{fail}}$ and add a 1 on top of the stack.
	For $q_{\text{sync}}$ we set $\delta(q_{\text{sync}}, a, \bot) = (q^A_0, \bot)$ and for all other transitions, we stay in $q_{\text{sync}}$ and leave the stack unchanged.
	With all other not yet defined transition in $Q^j$, we go to $q^j_{\text{fail}}$ and add a 1 on top of the stack.
	
	Following previous arguments, it is clear that any synchronizing word $w$ for $M$ in the same stack model must start with the letter $a$ and contain at at least one sequence $\#\#$. Further, let $j \in w[|w|]$ with $w[j] = \#$ be the position of the first occurrence of a sequence $\#\#$ and let $i \in w[|w|]$ with $w[i] = a$ be the position of the last occurrence of $a$ before position $j$. Then, the sub-word $w[i+1\, ..\, j-1]$ describes a solution for the PCP instance. Conversely, a solution of the PCP can be embedded in a synchronizing word  for $M$.
\end{proof}

The following proof gives proof details that do not completely follow the proof sketch of the main text, but are rather showing the power of the problem \textsc{Sync-From-Into-Subset} as a means for showing \PSPACE-membership.

\begin{proof}[Proof of Theorem~\ref{thm:DPDA-0-turn-PSPACE-Membership}]
	We prove the claims for DPDAs, as DCAs are a subclass hereof.
	Let $M = (Q , \Sigma, \Gamma, \delta, \bot)$ be a DPDA.
	We first focus on \textsc{0-Turn-Sync-DPDA-Empty}.
	As we need to synchronize with an empty stack, the 0-turn condition forbids us to write anything on the stack at all. Hence, we only keep in $M$ transitions of the form $\delta(q, \sigma, \bot) = (q', \bot)$ for $q, q' \in Q$, $\sigma\in \Sigma$ and delete all other transitions from $M$. We call the obtained automaton $M'$. We may observe that $M'$ is basically a partial DFA and the problem of synchronizing $M$ with a 0-turn synchronizing word in the empty stack model has been reduced to synchronizing $M'$ without using an undefined transition. The latter problem is called the \textsc{Careful Synchronization} problem and is solvable in \PSPACE~\cite{DBLP:journals/mst/Martyugin14}.
Notice that the argument given in the proof sketch in the main text can also be seen as an alternative proof (or, it can be easily adapted to this end)
of \PSPACE-membership of   \textsc{Careful Synchronization}.
	
	For the problem \textsc{0-Turn-Sync-DPDA-Arb} it is sufficient for each run to only keep the symbol on top of the stack in memory, as we are not allowed to decrease the height of the stack at any time, but transitions might still depend on the symbol on top of the stack. As we have no restriction on the stack content for synchronization, we can safely forget all other symbols on the stack. 
Hence, we can construct from $M$ a partial DFA $M'$ with state set $Q \times \Gamma$ by re-interpreting transitions $\delta(q, \sigma, \gamma) = (q', \gamma')$ for $q, q' \in Q, \sigma\in\Sigma, \gamma\in \Gamma, \gamma' \in \Gamma^*$ as $\delta((q, \gamma), \sigma) = (q', \gamma'[|\gamma'|])$ and deleting transitions of the form $\delta(q, \sigma, \gamma) = (q', \epsilon)$.
The problem of synchronizing $M$ with a 0-turn synchronizing word in the arbitrary stack model has now been reduced to finding a word that brings all states $(q', \bot)$ of $M'$ with $q' \in Q$ into one of the sets $S_q=\{(q, \gamma)\mid \gamma \in \Gamma\}$ for $q \in Q$ without using an undefined transition.
	We can solve this problem using polynomial space by guessing a path in the $|Q|$-fold product automaton $M'^{|Q|}$. The automaton $M'^{|Q|}$ with state set $(Q \times \Gamma)^{|Q|}$, consisting of $|Q|$-tuples of states in $Q \times \Gamma$, and alphabet $\Sigma$, is defined based on $M'$ by simulating the transition function $\delta$ of $M'$ on every single state in a $|Q|$-tuple state of $M'^{|Q|}$ independently, yielding the transition function $\delta^{|Q|}$ of $M'^{|Q|}$. Here, $\delta^{|Q|}$ is only defined on a $|Q|$-tuple if and only if $\delta$ is defined on every state of this tuple. Clearly, the size of $M^{|Q|}$ is $\mathcal{O}((|Q||\Gamma|)^{|Q|})$ and we can guess a path from the state $((q_1, \bot), (q_2, \bot), \dots, (q_n, \bot))$ for $Q= \{q_1, q_2, \dots, q_n\}$ to one of the sets $S_q$ for $q \in Q$ using space $\mathcal{O}(\log(|Q||\Gamma|)|Q|)$. We conclude the proof with Savitch's famous theorem proving $\PSPACE=\NPSPACE$~\cite{DBLP:journals/jcss/Savitch70}.
%
%
%
\end{proof}

\begin{proof}[Details of the construction of the transition function of $M_q$ in the proof of Theorem~\ref{thm:DCA-1-turn-PSPACE-Membership}]
We now explain how to build the transition function of $M_q$ from the specification of $M$.
\begin{itemize}
\item Whenever there is a transition $\delta(p,a,\bot)=(p',\bot)$ in $M$, $M_q$ can transition from $p$ to $p'$ upon reading $a$, leaving its counter untouched.
Formally, this means that  $\delta_q((p,1,0),a,\bot)=((p',1,1),\bot)$ in $M_q$. Moreover, for any state $r\in Q$ and any letter $b\in\Sigma$,  $\delta_q((r,1,1),b,\bot)=((r,1,0),\bot)$ in $M_q$.
Notice that in this stage one, $M_q$ knows that the counter is zero. In particular, any word read in this stage could be accepted.
\item If there is a transition  $\delta(p,a,\bot)=(p',\bot 1^\ell)$ with $\ell>0$ in $M$, $M_q$ moves into stage two. Hence, there is a  transition  $\delta_q((p,1,0),a,\bot)=((p',2,1),\bot 1^\ell)$ in $M_q$.
This certifies that we have truly entered a proper upstroke phase. 
\item We can check that we are in the proper upstroke phase by defining the transitions as $\delta_q((r,2,1),b,1)=((r,2,0),1)$ in $M_q$ for any $r\in Q$ and any $b\in\Sigma$. 
\item Transitions of the form  $\delta(p,a,1)=(p', 1^\ell)$ with $\ell> 0$ in $M$ let $M_q$ stay in stage two. Hence, there is a  transition  $\delta_q((p,1,0),a,1)=((p',2,1), 1^\ell)$ in $M_q$.
\item However, transitions of the form  $\delta(p,a,1)=(p', \varepsilon)$  in $M$ let $M_q$ move in stage three. We have now   truly entered a proper downstroke phase. Hence, there is a  transition  $\delta_q((p,2,0),a,1)=((p',3,1), \varepsilon)$ in $M_q$.
\item We can check that we are in the proper downstroke phase by setting the transition function as $\delta_q((r,3,1),b,1)=((r,3,0),1)$ in $M_q$ for any $r\in Q$ and any $b\in\Sigma$. 
\item We stay within stage three by introducing $\delta_q((p,3,0),a,1)=((p',3,1), 1^\ell)$ in $M_q$ for rules  $\delta(p,a,1)=(p', 1^\ell)$ with $\ell\in\{0,1\}$.
\item We will move into stage four once we have arrived at an empty stack again. This is realized by  having transitions $\delta_q((r,3,1),b,\bot)=((r,4,0),\bot)$ in $M_q$ for any $r\in Q$ and any $b\in\Sigma$. 
\item Whenever there is a transition $\delta(p,a,\bot)=(p',\bot)$ in $M$, $M_q$ can transition from $p$ to $p'$ upon reading $a$, leaving its counter untouched, i.e.,  $\delta_q((p,4,0),a,\bot)=((p',4,1),\bot)$ in $M_q$. Moreover, for any state $r\in Q$ and any letter $b\in\Sigma$,  $\delta_q((r,4,1),b,\bot)=((r,4,0),\bot)$ in $M_q$. Notice that in this stage four, $M_q$ knows again that the counter is zero. In particular, any word read in this stage could be accepted.
\end{itemize} 
\end{proof}

\newpage

\section{Appendix: Discussing Blind Counter Automata}

Recall that a counter automaton is blind if it has a counter that stores some integer (that can also be negative), but it
can test emptiness only at the very end of its computation, being hence part of the definition of the accepted language.
If we want to model this behavior similar to our previous definitions, we would therefore consider an element from $Q\times\mathbb
Z$ as a blind counter automaton configuration. The transition function $\delta$ would map $Q\times \Sigma$ to $Q\times \{-1,0,1\}$, meaning that if $\delta(q,\sigma)=(q',z)$, then $(q,c)\mto[\sigma] (q',c+z)$. The remainder of this formalization is standard and hence omitted.
Alternatively, we could try to  formalize this behavior also as a pushdown automaton. But then, the automaton must store in an additional bit if the counter stores a positive or a negative number. This also means that depending on this bit of information, incrementing the counter could either mean pushing~$1$ onto the stack or popping~$1$ from the stack, and (in a reversed fashion), this is also true when decrementing the counter. Clearly, this additional bit of information cannot be tested by the machine itself upon processing, as this would destroy the blindness condition; yet, this bit influences the formal processing, so that any formalization of blind counter automata (and their synchronizability) as special push-down automata looks somewhat strange.

How can we use this model for synchronization purposes?
First, observe that as the counter contents never influences the run of a blind automaton, synchronization with the arbitrary stack model would mean  just synchronization of the underlying DFA.

This looks different for the empty stack model. Here, we (have to) use a product automaton construction as in the proof of Theorem~\ref{thm:DPBCA-Sync} to reduce this synchronization problem to the emptiness problem of blind (real-time) deterministic multi-counter machines. By adding the possibility to decrement all counters upon reading a special symbol at the end, the emptiness problem of blind (real-time) deterministic multi-counter machines can also be used to show the decidability of the synchronization problem for deterministic blind counter automata. We summarize our observations as follows.

\begin{proposition}The problems \textsc{Sync-DBCA-Empty}, \textsc{Sync-DBCA-Same}, and \textsc{Sync-DBCA-Arb} are decidable, the latter even in polynomial time.
\end{proposition}

As a final remark, let us mention that already Greibach~\cite{greibach1978blind} observed the relations between (quasi-realtime) blind counter automata and reversal-bounded counter automata (reversals is just another name for turns), which also links to our discussion of finite-turn automata. However, when making a blind counter automaton reversal-bounded, the number of counters is increased, so that we cannot compare these models for a fixed number of  counters, which is of course our focus when studying language classes (and automata models) between regular and context-free.

\end{document}